\documentclass[a4paper,11pt]{amsart}
\usepackage{amssymb}
\usepackage{amsmath}
\usepackage{amsfonts}
\usepackage{pstricks,pst-plot,pst-node,psfrag,graphicx,pst-grad,amsfonts}
\usepackage{float}
\usepackage{multirow,enumerate,booktabs}
\newtheorem{theorem}{Theorem}
\usepackage{anysize}
\marginsize{1in}{1in}{1in}{1in}
\begin{document}
\bibliographystyle{amsplain}
\title{The inverse Lindley distribution: A stress-strength reliability model}
\author[Vikas Kumar Sharma]{Vikas Kumar Sharma}
\email{vikasstats@rediffmail.com}
\author[Sanjay Kumar Singh]{Sanjay Kumar Singh}
\author[Umesh Singh]{Umesh Singh}
\author[Varun Agiwal]{Varun Agiwal}
\address{Vikas Kumar Sharma, Sanjay Kumar Singh \& Umesh Singh
\newline \indent Department of Statistics and DST-CIMS,
\newline \indent Banaras Hindu University,
\newline \indent Varanasi-221005, India,
\newline \indent Correspondence to: vikasstats@rediffmail.com}
\address{Vikas Kumar Sharma \& Varun Agiwal
\newline \indent Department of Statistics,
\newline \indent Central University of Rajasthan, Bandersindri
\newline \indent Ajmer-00000, India}

\begin{abstract}
In this article, we proposed an inverse Lindley distribution and studied its fundamental properties such as quantiles, mode, stochastic ordering and entropy measure. The proposed distribution is observed to be a heavy-tailed distribution and has a upside-down bathtub shape for its failure rate. Further, we proposed its applicability as a stress-strength reliability model for survival data analysis. The estimation of stress-strength parameters and $R=P[X>Y]$, the stress-strength reliability has been approached by both classical and Bayesian paradigms. Under Bayesian set-up, non-informative (Jeffrey) and informative (gamma) priors are considered under a symmetric (squared error) and a asymmetric (entropy) loss functions, and a Lindley-approximation technique is used for Bayesian computation. The proposed estimators are compared in terms of their mean squared errors through a simulation study. Two real data sets representing survival of Head and Neck cancer patients are fitted using the inverse Lindley distribution and used to estimate the stress-strength parameters and reliability.  \vspace{.5cm}
\end{abstract}
\maketitle
\section{Introduction}
\cite{Lindley58} introduced a mixture of exponential$(\theta)$ and gamma$(2,\theta)$ distributions with mixing proportion is $\left[\theta/(\theta+1)\right]$, in the context of Bayesian statistics, as a counter example of fudicial statistics. This mixture is called the Lindley distribution and defined by the following probability density and survival functions,
\begin{equation}
f\left(y;\theta\right)=\frac{\theta^2}{\left(1+\theta \right) }\left(1+y \right) e^{-\theta y}, y,\theta > 0 
\label{eq1}
\end{equation}
\begin{equation}
S\left(y;\theta\right)=1- F\left(y;\theta \right)=1-\left( 1+\frac{\theta y}{1+\theta}\right) e^{-\theta y},\theta>0, y>0 
\label{eq2}
\end{equation}
\cite{Ghitny08} investigated the mathematical and statistical properties of Lindley distribution. They have shown that this probability model can be a better model than the well-known exponential distribution in some particular cases.
Since then, this distribution is becoming increasing popular for modeling lifetime data and has been discussed by many authors in different context such as censoring (\cite{Krishna11,SinghSharma14,Singh13}), stress-strength model (\cite{Mutairi13}), load-sharing system model (\cite{SinghGupta12}), competing risk model (\cite{Mazucheli11}) and Bayes estimation (\cite{Sajid}).

Since, the Lindley distribution is only applicable of modelling to the monotonic increasing hazard rate data, its applicability is restricted to the data that show non-monotone shapes (bathtub and upside-down bathtub) for their hazard rate. Though various article have found in literature that address the analysis of bathtub shape data, the limited attention has been paid to the analysis of upside-down bathtub shape data. Very recently, \cite{Sharma14} proposed a lifetime model with upside-down bathtub shape hazard rate function that is capable of modelling many real problems, for example failure of washing machines, survival of Head and Neck cancer patients, and survival of patients with breast cancer, see the references cited in \cite{Sharma14}.

Considering the fact that all inverse distribution possess the upside-down bathtub shape for their hazard rates, we, in this article, proposed a inverted version of the Lindley distribution that can be effectively used to model the upside-down bathtub shape hazard rate data. If a random variable $Y$ has a Lindley distribution LD($\theta$), then the random variable $X=(1/Y)$ is said to be follow the inverse Lindely distribution having a scale parameter $\theta$ with its probability density function (Pdf), defined by
\begin{equation}
f(x;\theta)=\frac{\theta^{2}}{1+\theta}\left( {\frac{1+x}{x^{3}}}\right){e^{\frac{-\theta}{x}}};x>0,\theta>0
\label{eq3}
\end{equation}
It is denoted by ILD$(\theta)$. The cumulative distribution function (Cdf) of inverse Lindley distribution is given by
\begin{equation}
F(x;\theta)=\left[ 1+\frac{\theta}{1+\theta}{\frac{1}{x}}\right]{e^{\frac{-\theta}{x}}};x>0,\theta>0
\label{eq4}
\end{equation}
Since this continuous distribution has the nice closed form expressions for the Cdf, hazard function as well as stress-strength reliability, its pertinency for survival analysis can never be denied in the literate.

The aim of this paper is two fold. The fist aim is to study the properties of the inverse Lindley distribution. The second aim of this paper is to develop the inferential procedure of the stress-strength reliability $R = P(Y < X)$, when $X$ represents the strength and $Y$ denotes the stress, both are independent inverse Lindley ($\theta_{1}$) and inverse Lindley ($\theta_{2}$), random variables respectively. The stress-strength reliability plays a very important role in reliability analysis and has nice probabilistic interpretation. As stated by \cite{Birnbaum13}, $R = P(Y < X)$ is the probability that failure will occur because, due to chance, a component with relatively low strength was paired off with a high stress. Many authors developed the estimation procedures for estimating the stress-strength reliability under both classical and Bayesian set-up, see \cite{Mutairi13,Ghitany13} and references cited therein.

The rest of the paper has been organized in the following sections. Section 2 provides comprehensive mathematical treatments to study the statistical properties of inverse Lindley distribution. The properties studied include: shapes of the
probability density function and the hazard rate function, mode, median, quantile function stochastic orderings, Reiny entropy, order statistics and their moments, and stress-strength reliability measures. The maximum likelihood estimators of the unknown parameters and stress-strength reliability along with their asymptotic distributions are presented in Section 3. Bayes estimators obtained by the method of Lindley approximation are discussed under Jeffrey and gamma priors using a symmetric, squared error loss and a symmetric, entropy loss functions in section 4. Section 5 presents the simulation study to study the behaviour of the proposed estimators. Section 6 illustrate an application by using the real data sets. Section 7 provides final concluding remarks.

\section{Properties of inverse Lindley distribution}
\subsection{Shape of density and hazard functions}
The first derivative of (\ref{eq3}) is given by
\begin{align*}
\frac{d}{dx}f(x)=-\frac{\theta^{2}}{1+\theta}\frac{e^{-\frac{\theta}{x}}}{x^{5}}\left(2x^{2}-\left(\theta-3 \right)x-\theta \right)
\end{align*}
and $\frac{d}{dx}f(x)\vert_{x=\mathcal{M}_{0}}=0$, where $\mathcal{M}_{0}$ is the mode of an inverse Lindley random variable and is given by
\begin{equation*}
\mathcal{M}_{0}=\frac{\theta-3+\sqrt{\left(\theta-3 \right)^{2}+8\theta }}{4}
\label{}
\end{equation*}
Figure (\ref{fig1}) shows the various shapes of inverse Lindley density for different choices of $\theta$ and also indicates that the density function of inverse Lindley density is uni-model in $x$.
The hazard function of ILD$(\theta)$ can be readily obtained as
\begin{equation}
h_{X}(x)=\dfrac{f_{X}(x)}{1-F_{X}(x)} =\frac{\theta^{2}\left(1+x\right)}{x^{2}\left[\theta+x\left(1+\theta\right)\left(e^{-\frac{\theta}{x}}-1 \right)\right] }
\end{equation}
\begin{figure}
\centering \includegraphics[scale=0.75]{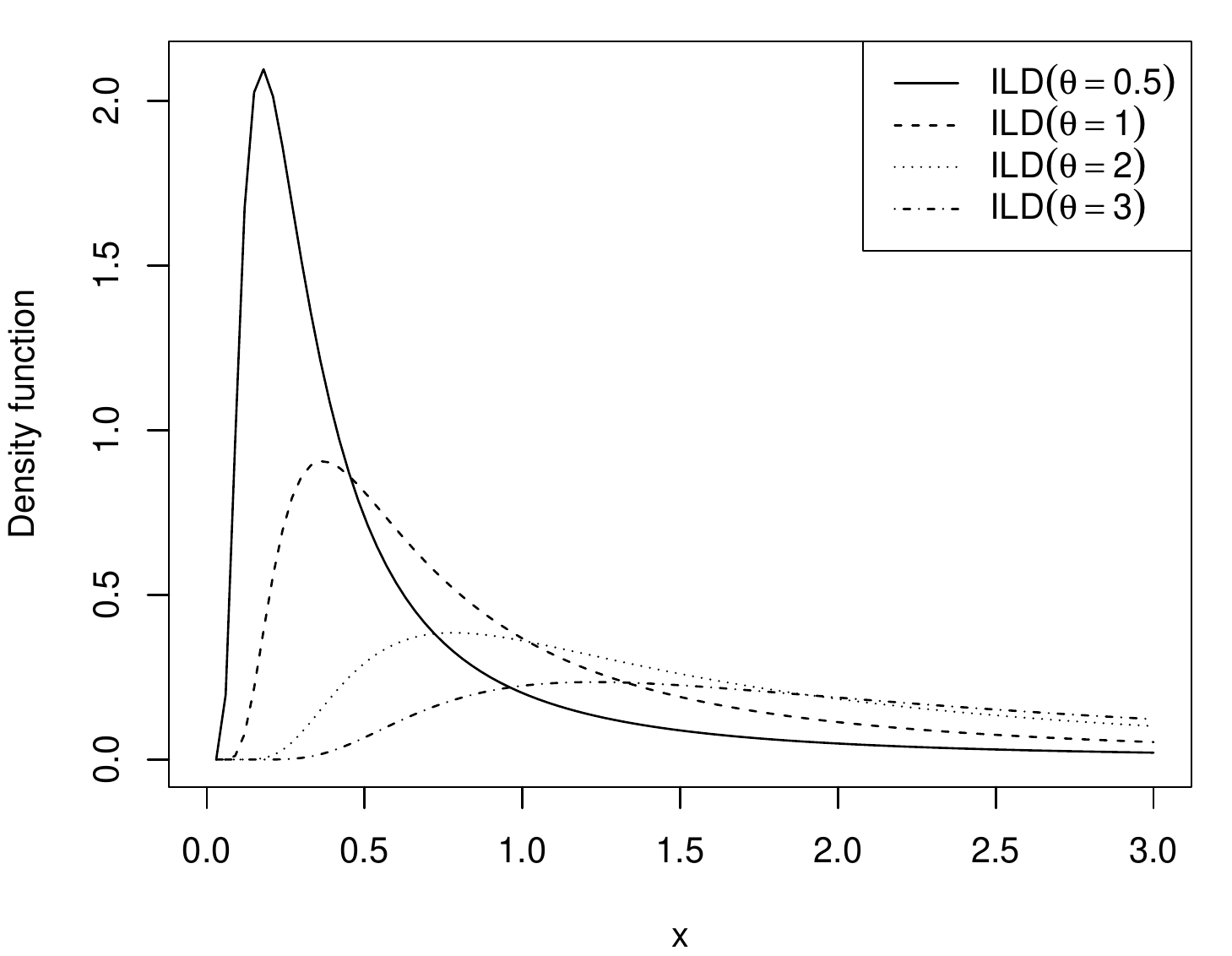}
\caption{Probability density functions of ILD$(\theta)$ for different values of $\theta$.}
\label{fig1}
\end{figure}
\begin{figure}
\centering\includegraphics[scale=0.75]{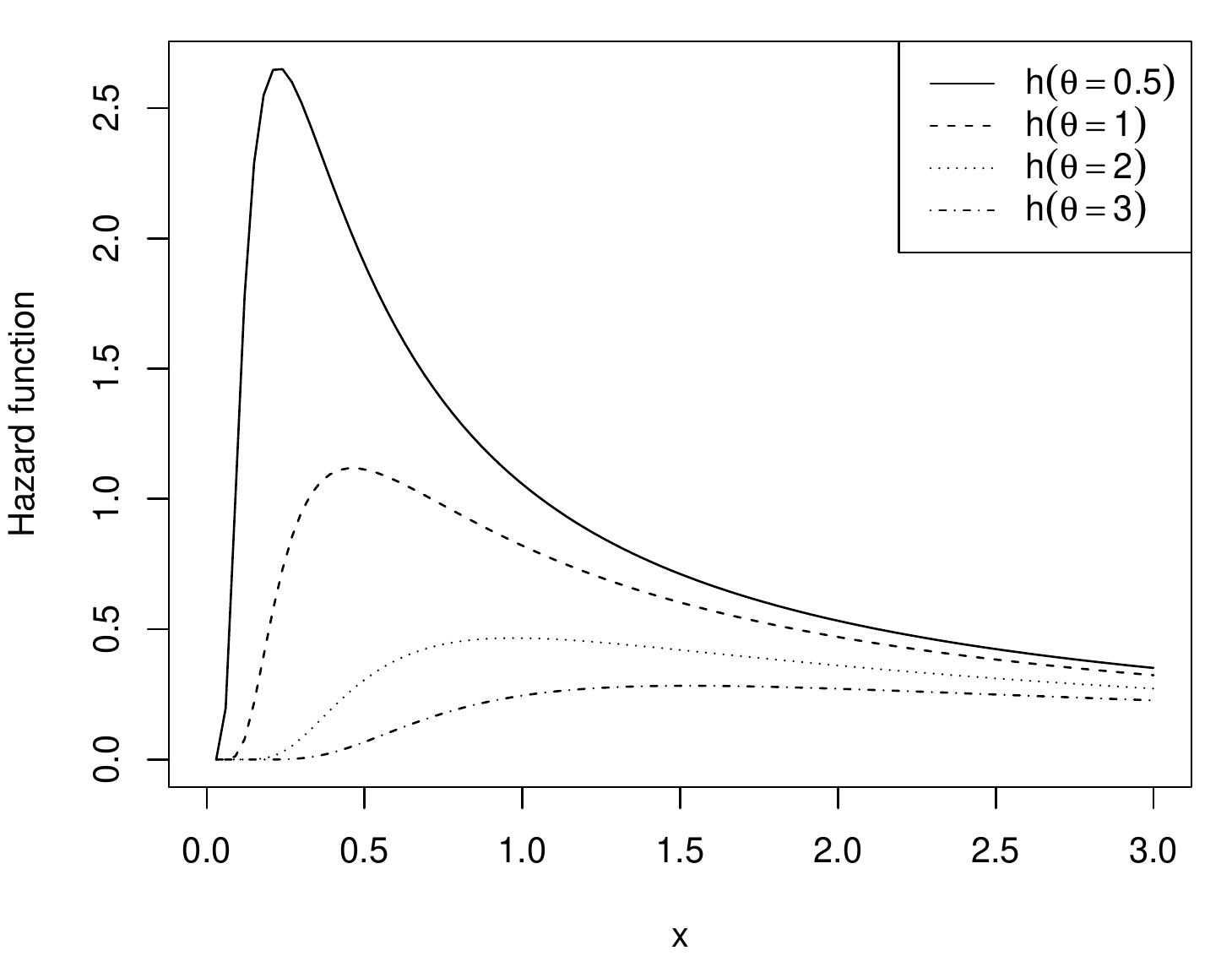}
\caption{Hazard functions of ILD$(\theta)$ for different values of $\theta$.}
\label{fig2}
\end{figure}
The monotonicity of hazard function of ILD$(\theta)$, for different values of the parameter $\theta$, has been showed graphically in Figure (\ref{fig2}). Clearly, the hazard function of inverse Lindley distribution is also a uni-model function in $x$ and achieve its maximum value at epoch, $x_{0}$, which can be obtained as the solution of the following equation
\begin{equation}
\left(1+\theta \right)e^{\frac{\theta}{x_{0}}}\left[\theta+\left(\theta-3 \right)x_{0}-2x^{2}_{0}\right]+x_{0}\left[4\theta+2\left(\theta+1 \right)x_{0}+3  \right]+2\theta=0
\end{equation}
\subsection{Quantile function}
Though formulas for the mean and the variance are difficult to obtain explicitly since the explicit algebraic expressions for involved integrals are not available, the quantiles are easy to evaluate. Let $T$ be an arbitrary random variable with cumulative distribution function $F_{T}(t)$=$P\left(T\le t\right)$, where t$\in \Re$. For any $p \in (0,1)$, the $p$th quantile $Q\left(p \right)$ of T is usually obtained as
\begin{equation}
F_{T}\left(Q\left(p \right)\right)=p
\label{eq:qp}
\end{equation}
\begin{theorem}
For any $\theta>0$, the quantile function of the inverse Lindley distribution $Q\left(p\right)$ is
\begin{equation}
Q\left(p\right)=-\left[1+\frac{1}{\theta}+\frac{1}{\theta}W_{-1}\left( -p(1+\theta)e^{-\left(1+\theta \right) }\right)\right]^{-1};~p\in (0,1),
\label{eq5}
\end{equation}
where $W_{-1}$ denotes the negative branch of the Lambert W function.
\end{theorem}
\begin{proof}
For any fixed $\theta> 0$, let $p\in (0, 1)$, we have to solve the equation (\ref{eq:qp}) with respect to $Q\left(p \right)$, for $Q\left(p \right)> 0$. From $(\ref{eq3})$, we have to solve the following equation:
\begin{equation}
\left(1+\theta+\frac{\theta}{Q\left(p \right)} \right) e^{-\frac{\theta}{Q\left(p \right)}}=\left(1+\theta \right)p 
\label{eq:qp1}
\end{equation}
Multiplying both side by $-e^{-(1+\theta)}$ in $(\ref{eq:qp1})$, we obtain
\begin{equation}
-\left(1+\theta+\frac{\theta}{Q\left(p \right)} \right) e^{-\left(1+\theta +\frac{\theta}{Q\left(p \right)}\right) }=-u\left(1+\theta \right)e^{-(1+\theta)}
\label{eq6}
\end{equation}
From $(\ref{eq6})$ together with $(\ref{eq3})$, we see that $-\left(1+\theta+\frac{\theta}{Q\left(p \right)} \right)$ is the Lambert W function of the real argument $-p\left(1+\theta \right)e^{-(1+\theta)}$ (The Lambert W function is a multivalued complex function defined as the solution of the equation $W(z)\exp\left(W(z)\right)=z$, where z is a complex number, see \cite{jorda10}). Then, we have
\begin{equation}
W\left(-u\left(1+\theta \right)e^{-(1+\theta)} \right) =-\left(1+\theta+\frac{\theta}{Q\left(p \right)} \right)
\label{eq7}
\end{equation}
Moreover, for any $\theta> 0$ and $x > 0$, it is immediate that $\left(1+\theta+\frac{\theta}{Q\left(p \right)} \right)>1$ and it can also be checked that $-p\left(1+\theta \right)e^{-(1+\theta)}\in(−\infty, -1)$ since $0<p<1$. Therefore, by taking into account the properties of the negative branch of the Lambert W function, $(\ref{eq7})$ becomes
\begin{equation}
W_{-1}\left(-p\left(1+\theta \right)e^{-(1+\theta)} \right) =-\left(1+\theta+\frac{\theta}{Q\left(p \right)} \right)
\label{eq8}
\end{equation}
which in turn implies the result. 
i.e.
\begin{equation}
Q\left(p \right)=-\left[1+\frac{1}{\theta}+\frac{1}{\theta}W_{-1}\left(-p(1+\theta)e^{-\left(1+\theta \right) }\right)\right]^{-1}
\end{equation}
\end{proof} 
For $p=0.5$, the median of inverse Lindley distribution is
\begin{equation}
\text{Median}(\tilde{\mu})= -\left[1+\frac{1}{\theta}+\frac{1}{\theta}W_{-1}\left( -0.5(1+\theta)e^{-\left(1+\theta \right) }\right)\right]^{-1}
\label{eq9}
\end{equation} 
\subsection{Stochastic ordering}
A random variable $X$ is said to be stochastically greater $\left(Y\leq_{st} X \right)$ than $Y$ if $F_{X}(x)\leq F_{Y}(x)$ for all $x$. In the similar way, $X$ is said to be stochastically greater $\left(X\leq_{st} Y \right)$ than $Y$ in the
\begin{enumerate}[(i)]
\item stochastic order$\left(X \leq_{st}Y \right)$ if $F_{X}(x) \geq F_{Y}(x)$ for all x.
\item hazard rate order $\left(X \leq_{hr}Y \right)$ if $h_{X}(x) \geq h_{Y}(x)$ for all x.
\item mean residual life order $\left(X \leq_{mrl}Y \right)$ if $m_{X}(x) \geq m_{Y}(x)$ for all x.
\item likelihood ratio order $\left(X \leq_{lr}Y \right)$ if $\left(\frac{f_{X}(x)}{f_{Y}(x)} \right) $ decreases in x.
\end{enumerate}
The following implications (\cite{ordering}) are well known
\begin{center}
$\left(X \leq_{lr}Y \right)\implies \left(X \leq_{hr}Y \right) \implies \left(X \leq_{mrl}Y \right)$\\
$\Downarrow$ \\
$ \left(X \leq_{st}Y \right)$
\end{center}
The inverse Lindley distributions are ordered with respect to the strongest likelihood ratio ordering as shown in the following theorem.
\begin{theorem}
Let $X \sim ILD(\theta_{1})$ and $Y \sim ILD(\theta_{2})$. If $\theta_{1}>\theta_{2}$, then $\left(X \leq_{lr}Y \right)$ and hence $\left(X \leq_{hr}Y \right), \left(X \leq_{mrl}Y \right)$ and $\left(X \leq_{hr}Y \right)$.
\end{theorem}
\begin{proof}
 First note that 
\begin{equation}
\begin{split}
\frac{f_{X}(x)}{f_{Y}(x)}&=\frac{\theta_{1}^{2}\left( 1+\theta_{2}\right)e^{-\frac{\theta_{1}}{x}}}{\theta_{2}^{2}\left( 1+\theta_{1}\right)e^{-\frac{\theta_{2}}{x}}} \\
&=\frac{\theta_{1}^{2}\left( 1+\theta_{2}\right)}{\theta_{2}^{2}\left( 1+\theta_{1}\right)}e^{-\frac{\left(\theta_{1}-\theta_{2}\right) }{x}}; x>0\\
\end{split}
\end{equation}
Since, for $\theta_{1}>\theta_{2}$
\begin{equation}
\frac{d}{dx}\frac{f_{X}(x)}{f_{Y}(x)}=\frac{\left(\theta_{1}-\theta_{2}\right)}{x^{2}}\frac{f_{X}(x)}{f_{Y}(x)}>0
\label{eq11}
\end{equation}
that is $\eqref{eq11}$ is a positive monotone decreasing function, hence  
$\frac{f_{X}(x)}{f_{Y}(x)}$ is decreasing in x. That is $X \leq_{lr}Y $. The remaining statements follow from the above relationship.
\end{proof}
\subsection{Entropy measure}
Entropy measures provide important tools to indicate variety in distributions at particular moments in time and to analysis evolutionary processes over time. This measures the variation of the uncertainty in the distribution of a random variable $X$. For a probability distribution, \cite{ref8} gave an expression of the entropy function, so called Renyi entropy, defined by
\begin{equation}
\tau_{R}(\gamma)= \frac{1}{1-\gamma}\log\left\lbrace\int f^{\gamma}(x)dx \right\rbrace \
\end{equation}
where $\gamma>0$ and $\gamma \neq 1$.  For the inverse Lindley distribution in $(\ref{eq2})$, note that
\begin{equation}
\tau_{R}(\gamma)=\frac{1}{1-\gamma}\log \int\limits_{0}^{\infty}\frac{\theta^{2\gamma}}{\left(1+\theta\right)^{\gamma} }\left({\frac{\left( 1+x\right)^{\gamma} }{x^{3\gamma}}}\right){e^{\frac{-\theta\gamma}{x}}}dx
\end{equation}
We known, $(1+z)^{j}$=$\sum\limits_{j=0}^{\infty}\binom{\gamma}{j} z^{j}$ and $\int\limits_{0}^{\infty}e^{-b/x}x^{-a-1}dx=\Gamma\left(a\right)/b^{a}$, then 
\begin{equation}
\begin{split}
\tau_{R}(\gamma)&=\frac{1}{1-\gamma}\log\left[ \frac{\theta^{2\gamma}}{\left(1+\theta\right)^{\gamma}}\sum\limits_{j=0}^{\infty}
\binom{\gamma}{j} \int\limits_{0}^{\infty}\frac{e^{\frac{-\theta\gamma}{x}}}{x^{3\gamma-j}}dx\right] \\
&=\frac{1}{1-\gamma}\log\left[ \frac{\theta^{2\gamma}}{\left(1+\theta\right)^{\gamma}}\sum\limits_{j=0}^{\infty}\binom{\gamma}{j} \frac{\Gamma\left({3\gamma-j-1}\right) }{\left(\theta\gamma\right) ^{3\gamma-j-1}}\right] 
\end{split}
\end{equation}
\section{Stress-strength reliability and maximum likelihood estimation}
Let Y and X be independent stress and strength random variables follow inverse Lindley distribution with parameter $\theta_{1}$ and $\theta_{2}$ respectively. Then, the stress-strength reliability $R$ is define as
\begin{equation}
\begin{split}
R&=P\left[Y<X \right]=\int\limits_{0}^{\infty}P\left[Y<X|X=x\right]f_{X}(x)dx=\int\limits_{0}^{\infty}f\left( x,\theta_{1} \right)F\left( x,\theta_{2} \right)dx\\
&=\int\limits_{0}^{\infty}\left[\left(1+\left(\frac{\theta_{2}}{1+\theta_{2}}\frac{1}{x} \right)  \right) e^{-\frac{\theta_{2}}{x}} \right]\left(\frac{\theta_{1}^{2}}{1+\theta_{1}} \right)\left( \frac{1+x}{x^{3}}\right)e^{-\frac{\theta_{1}}{x}} dx\\
&=\left(\frac{\theta_{1}^{2}}{1+\theta_{1}} \right)\int\limits_{0}^{\infty}\left( \frac{1+x}{x^{3}}\right)e^{-\frac{\theta_{1}+\theta_{2}}{x}}dx + \frac{\theta_{1}^{2}\theta_{2}}{\left(1+\theta_{1} \right)\left( 1+\theta_{2}\right)}\int\limits_{0}^{\infty}\frac{1}{x} \left( \frac{1+x}{x^{3}}\right)e^{-\frac{\theta_{1}+\theta_{2}}{x}}dx\\
&=\left(\frac{\theta_{1}^{2}}{1+\theta_{1}} \right)\int\limits_{0}^{\infty}\left( \frac{1+x}{x^{3}}\right)e^{-\frac{\theta_{1}+\theta_{2}}{x}}dx + \frac{\theta_{1}^{2}\theta_{2}}{\left(1+\theta_{1}\right)\left( 1+\theta_{2}\right)}\left[\int\limits_{0}^{\infty}x^{-4}e^{-\frac{\theta_{1}+\theta_{2}}{x}}dx+\int\limits_{0}^{\infty}x^{-3}e^{-\frac{\theta_{1}+\theta_{2}}{x}}dx \right] 
\end{split}
\end{equation}
Using the definitions of ILD and inverse gamma distributions, we get the expression for stress-strength reliability as
\begin{equation}
\begin{split}
R&=\left(\frac{\theta_{1}^{2}}{1+\theta_{1}}\right)\left( \frac{1+\theta_{1}+\theta_{2}}{\left(\theta_{1}+\theta_{2}\right)^{2}}\right)+ \frac{\theta_{1}^{2}\theta_{2}}{\left(1+\theta_{1}\right)\left( 1+\theta_{2}\right)}\left[\frac{\Gamma 3}{\left(\theta_{1}+\theta_{2}\right)^{3}}+\frac{\Gamma 2}{\left(\theta_{1}+\theta_{2}\right)^{2}}\right] \\
R&=\frac{\theta_{1}^{2}\left[ \left(\theta_{1}+\theta_{2} \right)^{2}\left(1+\theta_{2} \right)+\left( \theta_{1}+\theta_{2}\right)\left( 1+2\theta_{2}\right)+2\theta_{2} \right] }{\left(1+\theta_{1} \right)\left( 1+\theta_{2}\right)\left( \theta_{1}+\theta_{2}\right)^{3}}
\end{split}
\label{eq:R}
\end{equation}
Since $R$ in (\ref{eq:R}) is a function of stress-strength parameters $\theta_{1}$ and $\theta_{2}$, we need to obtain the maximum likelihood estimators (MLEs) of $\theta_{1}$ and $\theta_{2}$ to compute the MLE of $R$ under invariance property of the MLE.

Suppose $X_{1},X_{2},...,X_{n}$ is a strength random sample from ILD$\left( \theta_{1}\right) $ and $Y_{1},Y_{2},...,Y_{m}$ is a stress random sample from ILD$\left( \theta_{2}\right) $ distributions. Thus, the likelihood function based on the observed sample is given by 
\begin{equation}
\begin{split}
L(\theta_{1},\theta_{2}\vert \textbf{\underbar x}, \textbf{\underbar y})
&=\frac{\theta_{1}^{2n}\theta_{2}^{2m}}{\left( 1+\theta_{1}\right)^{n}\left( 1+\theta_{2}\right)^{m} }\prod\limits_{i=1}^{n}\left({\frac{1+x_{i}}{x_{i}^{3}}}\right)\prod\limits_{j=1}^{m}\left({\frac{1+y_{j}}{y_{j}^{3}}}\right){\exp\left\lbrace -\left(\theta_{1} s_{1}+ \theta_{2} s_{2}\right) \right\rbrace }
\end{split}
\end{equation}
where, $s_{1}=\sum\limits_{i=1}^{n}\frac{1}{x_{i}}$ \text{~and~} $s_{2}=\sum\limits_{j=1}^{m}\frac{1}{y_{j}}$. The log-likelihood function is given by
\begin{equation}
\begin{split}
\log L(\theta_{1},\theta_{2})&=2n\log\theta_{1}+2m\log\theta_{2}-n\log\left( 1+\theta_{1}\right)-m\log\left( 1+\theta_{2}\right) -\theta_{1} s_{1}-\theta_{2} s_{2}\\
&+\sum\limits_{i=1}^{n}\log\left( \frac{1+x_{i}}{x_{i}}\right)+\sum\limits_{j=1}^{m}\log\left( \frac{1+y_{j}}{y_{j}}\right)
\end{split}
\end{equation}
The MLE of $\theta_{1}$ and $\theta_{2}$, say $\hat{\theta}_{1}$ and $\hat{\theta}_{2}$ respectively, can be obtained as the solutions of the following equations
\begin{equation}
\frac{\partial\log L}{\partial\theta_{1}}=\frac{2n}{\theta_{1}}-\frac{n}{1+\theta_{1}}-s_{1}
\label{eq12}
\end{equation}
\begin{equation}
\frac{\partial\log L}{\partial \theta_{2 }}=\frac{2m}{\theta_{2}}-\frac{m}{1+\theta_{2}}-s_{2}
\label{eq13}
\end{equation}
From $\eqref{eq12}$ and $\eqref{eq13}$, obtain MLEs as 
\begin{equation}
\hat{\theta}_{1}=-\frac{\left(s_{1}-n \right)+\sqrt{\left(s_{1}-n \right)^{2}+8ns_{1} } }{2s_{1}}
\end{equation}
\begin{equation}
\hat{\theta}_{2}=-\frac{\left(s_{2}-m \right)+\sqrt{\left(s_{2}-m \right)^{2}+8ms_{2} } }{2s_{2}}
\end{equation}
Hence, using the invariance property of the MLE, the maximum likelihood estimator $\hat{R}_{mle}$ of R can be obtained by substituting $\hat{\theta}_{k}$ in place of $\theta_{k}$ for k = 1, 2.
\begin{equation}
\hat{R}_{mle}= \frac{\theta_{1}^{2}\left[ \left(\theta_{1}+\theta_{2} \right)^{2}\left(1+\theta_{2} \right)+\left( \theta_{1}+\theta_{2}\right)\left( 1+2\theta_{2}\right)+2\theta_{2} \right] }{\left(1+\theta_{1} \right)\left( 1+\theta_{2}\right)\left( \theta_{1}+\theta_{2}\right)^{3}}\Bigg\vert_{\theta_{k}=\hat{\theta}_{k},k=1,2}
\end{equation}
\subsection{Asymptotic distribution and confidence intervals for $\theta_{1}$, $\theta_{2}$ and $R$}
In this subsection, we derived the asymptotic distributions of $\hat{\theta}_{k}$ as well as of $\hat{R}_{mle}$. Based on the asymptotic distribution of $\hat{R}_{mle}$, we also obtained the asymptotic confidence interval for R. 
For large sample, an estimator is said to be asymptotic efficient for estimating $\theta_{k}$ if 
\begin{equation}
\sqrt{n}\left(\hat{\theta}_{k} -\theta_{k}\right)\rightarrow N\left( 0, I\left( \theta_{k}\right)^{-1} \right) 
\end{equation}
where, 
\begin{equation*}
I\left( \theta_{k}\right)=-E\left(\frac{\partial^{2}f(x)}{\partial\theta_{k}^{2}} \right) = \frac{\theta_{k}^{2}+4\theta_{k}+2}{\theta_{k}^{2}\left( 1+\theta_{k}\right)^{2}} ;~~k=1,2
\end{equation*}
By using delta method, the Fisher information matrix of $\Theta=(\theta_{1},\theta_{2})$ is defined as
$$I\left(\Theta\right)= -
\begin{pmatrix}
E\left(\frac{\partial^{2}L}{\partial\theta^{2}_{1}}\right) & E\left(\frac{\partial^{2}L}{\partial\theta_{1}\partial\theta_{2}}\right)\\
E\left(\frac{\partial^{2}L}{\partial\theta_{2}\partial\theta_{1}}\right)& E\left(\frac{\partial^{2}L}{\partial\theta^{2}_{2}}\right)
\end{pmatrix}=
\begin{pmatrix}
I_{11}& 0\\
0 & I_{22}
\end{pmatrix}$$
The asymptotic $100\left(1-\alpha \right)\%$ confidence intervals for $\theta_{k}, k=1,2$ are defined by
$$\left\lbrace \hat{\theta}_{k}\mp z_{\alpha/2}\sqrt{\frac{\hat{\theta}_{k}^{2}\left( 1+\hat{\theta}_{k}\right)^{2}}{\hat{\theta}_{k}^{2}+4\hat{\theta}_{k}+2} } \right\rbrace $$ 
where, $z_{\alpha/2}$ is the upper $(\alpha/2)$th percentile of a standard normal variable. Similarly, the asymptotic distribution of stress-strength reliability as $n \rightarrow\infty$ and $m\rightarrow\infty$ is given by
\begin{equation}
\frac{\hat{R}-R}{\sqrt{\frac{R_{1}^{2}}{nI_{11}}+\frac{R_{2}^{2}}{m I_{22}}}}\rightarrow N\left(0,1 \right) 
\end{equation}

where,
\begin{equation*}
\begin{split}
&R_{1}=\frac{dR}{d\theta_{1}}=\dfrac{\theta_{1}\theta_{2}^{2}\left[\theta_{1}^{3}+2\theta_{1}^{2}\left(\theta_{2}+3\right)+2\left(\theta_{2}^{2}+3\theta_{2}+3\right)+\theta_{1}\left(\theta_{2}+2\right)\left(\theta_{2}+6\right)\right] }{\left(1+\theta_{1}\right)^{2}\left( 1+\theta_{2}\right)\left( \theta_{1}+\theta_{2}\right)^{4}}\\ 
&R_{2}=\frac{dR}{d\theta_{2}}=\dfrac{-\theta_{1}^{2}\theta_{2}\left[6+\theta_{1}^{2}\left(\theta_{2}+2\right)+\theta_{2}\left(\theta_{2}^{2}+6\theta_{2}+12\right)+2\theta_{1}\left(\theta_{2}+1\right)\left(\theta_{2}+3\right)\right] }{\left(1+\theta_{1}\right)\left( 1+\theta_{2}\right)^{2}\left( \theta_{1}+\theta_{2}\right)^{4}} 
\end{split}
\end{equation*}
Although, $\hat{R}_{mle}$ can be obtained in explicit form, the exact distribution of $\hat{R}$ is difficult to obtain. Due to this reason, we construct the confidence interval for $R$ based on the asymptotic distribution of the maximum likelihood estimator $\hat{R}_{mle}$. Thus, asymptotic $100\left(1-\alpha \right)\%$ confidence interval for R can be easily obtained as 
\begin{equation*}
\left\lbrace  \hat{R}_{mle}\mp z_{\frac{\alpha}{2}}\sqrt{\frac{\hat{R}_{1}^{2}}{nI\left(\hat {\theta}_{1}\right)}+\frac{\hat{R}_{2}^{2}}{mI\left(\hat {\theta}_{2}\right)}} \right\rbrace 
\end{equation*}
where, $I\left(\hat{\theta}_{k}\right)$ and $\hat{R}_{k}, k= 1,~2$ are the MLE's of $I\left( \theta_{k}\right)$ and ${R}_{k}$, respectively.

\section{Bayes estimation}
In this section, we proposed the Bayesian estimators of the unknown parameters $\theta_{1}$, $\theta_{2}$ and stress-strength reliability function $R$. For Bayesian estimation, we need to specify prior distributions, $\pi\left( \theta \right)$ for the model parameters. Here, we consider a non-informative, \cite{JEFFREY61} and an informative gamma priors for the unknown parameters $\theta_{1}$ and $\theta_{2}$. The joint posterior distribution of $\theta_{1}$ and $\theta_{2}$ is defined by
\begin{equation}
\pi \left(\theta_{1},\theta_{2}\vert \textbf{data}\right)= K L\left( \theta_{1},\theta_{2}\vert \textbf{data}\right)  \pi\left( \theta_{1} \right)\pi\left( \theta_{2} \right)
\end{equation}
where, $K^{-1}=\int\limits_{0}^{\infty} \int\limits_{0}^{\infty}  L\left( \theta_{1},\theta_{2}\vert \textbf{data}\right) \pi\left( \theta_{1} \right)\pi\left( \theta_{2} \right) d\theta_{1} d\theta_{2}$.

\cite{Berger58} argued that if no information or less information is available regarding the population parameter, it is better to consider a non-informative prior for the unknown parameter. The joint prior under a methodology given by \cite{JEFFREY61}, is given by 
\begin{equation}
\pi_{j}(\theta_{1},\theta_{2})=\frac{\sqrt{\left( \theta_{1}^{2}+4\theta_{1}+2\right)\left( \theta_{2}^{2}+4\theta_{2}+2\right)}}{\theta_{1}\theta_{2}\left( 1+\theta_{1}\right)\left( 1+\theta_{2}\right) }
\end{equation}
We now suppose some information on the parameters $\theta_{1}$ and $\theta_{2}$ is available. Since both parameters of the inverse Lindely distribution are positive and greater than zero, it is assumed $\theta_{1}$ and $\theta_{2}$ have gamma prior distributions of the following forms.
\begin{equation}
\pi(\theta_{k})=\frac{b_{k}^{a_{k}}}{\Gamma a_{k}}\theta_{k}^{a_{k}-1}e^{-\theta_{k}b_{k}}, k=1,2.
\end{equation}
The joint prior is then given by
\begin{equation}
\pi_{g}(\theta_{1},\theta_{2})\propto\theta_{1}^{a_{1}-1}\theta_{2}^{a_{2}-1}e^{-\left( \theta_{1}b_{1}+\theta_{2}b_{2}\right) }
\end{equation}
Here, all hyper-parameters $a_{1},b_{1},a_{2},b_{2}$ are assumed to be known  that are chosen to reflect the prior information about parameters, based on a prior information available that fits the gamma distribution which is mainly subjective probabilistic. Note, when $a_{1}=b_{1}=a_{2}=b_{2}=0$  then we have non-informative and improper prior.

We, further, considered two loss functions, one is asymmetric known as quadratic loss or squared error loss and other one is asymmetric, entropy loss function. The Bayes estimate of any parametric function, say $\phi \left(\Theta \right)$ under squared error loss function (SELF) and entropy loss function (ELF) are defined by 
\begin{align}
\phi^{\star}_{self}\left(\Theta\vert\textbf{data}\right)&=E\left[\phi\left(\Theta\right)\vert\textbf{data}\right]=K \int\limits_{\Theta}\phi\left(\Theta\right) \pi \left(\Theta\vert \textbf{data}\right) d\Theta \\
\phi^{\star}_{elf}\left(\Theta\vert\textbf{data}\right)&=\left(E\left[\phi^{-1}\left(\Theta\right)\vert\textbf{data}\right] \right) ^{-1}=\left( K \int\limits_{\Theta}\phi^{-1}\left(\Theta\right)\pi \left(\Theta\vert \textbf{data }\right) d\Theta\right)^{-1} 
\end{align}
It is to be noticed here that a major difficulty in the implementation of Bayes procedure is the evaluation of the ratio of two integrals for which closed expression is not easy to obtain analytically. Therefore, we proposed the use of an approximation technique suggested by \cite{Lindley1980} to solve Bayesian problems. Many authors including (\cite{Singh2008}) have used this technique to obtain the Bayes estimates of the parameters of various lifetime distributions.
\subsection{Lindley-approximation}
Consider that the posterior expectation of $\phi\left(\Theta\vert\textbf{data}\right)$ is expressible in the form of ratio of integrals as given below
\begin{equation}
\phi^{\star}\left(\Theta\vert\textbf{data}\right)=\frac{\int\limits_{\Theta}\phi\left(\Theta\right)\exp\left[\log L\left(\Theta\right)+\rho\left(\Theta \right)\right]d\Theta}{\int\limits_{\Theta}\exp\left[\log L\left(\Theta\right)+\rho\left(\Theta\right)\right]d\Theta}
\label{LDR}
\end{equation}
where, \\
$\Theta=\left( \theta_{1},\theta_{2},\ldots,\theta_{p}\right)$,\\
$\phi\left(\Theta\right)$=  Parametric function of interest,\\
$\log L \left(\Theta\right)$= Log-likelihood function, and\\
$\rho\left(\Theta\right)$= Log of joint prior of $\Theta$, \\

For a sample large and under some conditions, the above equation (\ref{LDR}) can be approximated as
\begin{equation}
\phi^{\star}\left(\Theta\vert\textbf{data}\right) \approx \phi+\frac{1}{2}\sum\limits_{i=1}^{p}\sum\limits_{j=1}^{p}\left(\phi_{ij}+2\phi_{i}\rho_{j}\right)\sigma_{ij}+\frac{1}{2}\sum\limits_{i=1}^{p}\sum\limits_{j=1}^{p} \sum\limits_{k=1}^{p} \sum\limits_{l=1}^{p} L_{ijk}\phi_{l}\sigma_{ij}\sigma_{kl}
\end{equation}
where,
\begin{equation*}
\begin{split}
&p=2,\phi_{1}=\frac{\partial \phi\left( \theta_{1},\theta_{2}\right) }{\partial\theta_{1}},	\phi_{2}=\frac{\partial \phi\left( \theta_{1},\theta_{2}\right)}{\partial\theta_{2}},
\phi_{11}=\frac{\partial^{2}\phi(\theta_{1},\theta_{2})}{\partial\theta_{1}^{2}}, \phi_{12}=\phi_{21}=\frac{\partial^{2}\phi\left( \theta_{1},\theta_{2}\right)}{\partial\theta_{1}\partial\theta_{1}},\\
&\phi_{22}=\frac{\partial^{2}\phi\left( \theta_{1},\theta_{2}\right) }{\partial\theta_{2}^{2}},
L_{111}=\frac{\partial^{3}\log L\left( \theta_{1},\theta_{2}\right) }{\partial\theta_{1}^{3}}, L_{112}=\frac{\partial^{3}\log L\left( \theta_{1},\theta_{2}\right) }{\partial\theta_{1}^{2}\partial\theta_{2}}, 
L_{122}=\frac{\partial^{3}\log L\left( \theta_{1},\theta_{2}\right) }{\partial\theta_{1}\partial\theta_{2}^{2}},\\
&L_{222}=\frac{\partial^{3}\log L\left( \theta_{1},\theta_{2}\right) }{\partial\theta_{2}^{3}},
\rho_{1}=\frac{\partial\log \pi\left( \theta_{1},\theta_{2}\right) }{\partial\theta_{1}},\rho_{2}=\frac{\partial\log \pi\left( \theta_{1},\theta_{2}\right) }{\partial\theta_{2}}\\
&\text{ and }
\begin{pmatrix}
\sigma_{11}& \sigma_{12}\\
\sigma_{21}& \sigma_{22}
\end{pmatrix}
=\begin{pmatrix}
I_{11}& I_{12}\\
I_{21}& I_{22}
\end{pmatrix}^{-1}.
\end{split}
\end{equation*}
The Bayes estimators of stress-strength parameters $\hat{\theta}_{k}$, and $R$ under squared error loss function are given by 
\begin{equation}
\begin{split}
\theta^{*}_{k,self}=&\hat{\theta}_{k, mle}+ \hat{\rho}_{k}\hat{\sigma}_{kk}+ 0.5\left(\hat{L}_{kkk}\hat{\sigma}_{kk}^{2}\right) \\
=&\hat{\theta}_{k,mle}, k=1,~2.
\end{split}
\end{equation} 
\begin{equation}
R^{*}_{self}=\hat{R}+\frac{1}{2}\left[ \hat{\sigma}_{11}\left(\hat{R}_{11}+2\hat{R}_{1}\hat{\rho}_{1}\right)+\hat{\sigma}_{22}\left(\hat{R}_{22}+2\hat{R}_{2}\hat{\rho}_{2}\right)\right]\\+\frac{1}{2}\left[\hat{L}_{111}\hat{R}_{1}\hat{\sigma}_{11}^{2}+\hat{L}_{222}\hat{R}_{2}\hat{\sigma}_{22}^{2}\right]
\end{equation}
and under entropy loss function are defined as 
\begin{equation}
\theta^{*}_{k,elf}=\left[ \frac{1}{\hat{\theta}_{k}}+\frac{\hat{\sigma}_{kk}}{\hat{\theta}_{k}^{2}}\left(\frac{1}{\hat{\theta}_{k}}-\hat{\rho}_{k}\right)-\frac{\hat{L}_{kkk}\hat{\sigma}_{kk}^{2}}{2\hat{\theta}_{k}^{2}}\right] ^{-1}
\end{equation}
\begin{equation}
R^{*}_{elf}  = \left[ \hat{R}^{-1}+\frac{1}{2}\left\lbrace  \hat{\sigma}_{11}\left(\hat{R}_{11}^{-1}+2\hat{R}_{1}^{-1}\hat{\rho}_{1}\right)+\hat{\sigma}_{22}\left(\hat{R}_{22}^{-1}+2\hat{R}_{2}^{-1}\hat{\rho}_{2}\right)\right\rbrace \\+\frac{1}{2}\left(\hat{L}_{111}\hat{R}_{1}^{-1}\hat{\sigma}_{11}^{2}+\hat{L}_{222}\hat{R}_{2}^{-1}\hat{\sigma}_{22}^{2}\right)\right]^{-1} 
\end{equation} 
Note that all terms and derivatives written with a hat $(''~\hat{ }~'')$ are calculated by replacing $\theta_{1}$ and $\theta_{2}$ with their maximum likelihood estimates. To obtain the Bayes estimates under Jeffrey and gamma priors, we take $\rho=\log \pi_{j}\left(\theta_{1},\theta_{2}\right)$ and $\rho=\log \pi_{g}\left(\theta_{1},\theta_{2}\right)$ respectively. The derivatives used under Lindley approximation are shown in Appendix.
\section{Simulation study}
This section consists a simulation study to analysing the behaviour of the proposed estimators based on simulates sample. In simulation, one, using a machine instead of going in the field for data collection, generates the pseudo-random numbers to study the properties of the model, and save time as well as money. Here, we studied the behaviour of the stress-strength reliability on the basis of simulated sample with varying sample sizes and different combinations of the stress-strength parameters.

For this purpose, we need a simulation algorithm for generating random sample from inverse Lindley distribution.
The simplest method used for this purpose is an inversion method that utilizes a probability integral transformation. Though the probability integral transformation under ILD can not be applied explicitly, we, can apply Newton's method (to solve the equation $F(x)-u(0,1)=0$) or a Lambert W function as suggested by \cite{jorda10}. To simulate a random sample from ILD distribution, One can also use a fact that the ILD distribution is a mixture of an inverse exponential ($\theta$) and an inverse gamma (2, $\theta$) distributions. The algorithm is then have the following steps.
\begin{enumerate}[(i)]
\item Generate $U_{i} \sim~ Uniform(0, 1), i = 1,2,\dots, n$.
\item Generate $V_{i} \sim~ Inverse~ Exponential(\theta), i = 1,2,\dots, n$.
\item Generate $W_{i} \sim~ Inverse~ Gamma(2,\theta), i = 1,2,\dots, n$.
\item If $U_{i} \leq p = \theta/\left( \theta + 1\right)$, then set $X_{i} = V_{i}$, otherwise, set $X_{i} = W_{i}, i = 1,2,\ldots, n.$
\end{enumerate} 
For simulation, we took three different combinations of $\left( \theta_{1},\theta_{2}\right)$=(1,2),(1,1),(2,1) such that for which stress-strength reliability results in a value (0.25) small, (0.50) moderate and (0.75) large, respectively. Under gamma prior,  the hyper-parameters ate taken to be zero to make a gamma prior non-informative. The resulting prior is denoted by Gamma 1. As in a simulation study, the true values of the model parameters are known priory, one can use this information to elicit the assumed prior distributions. Here, we,  under this consideration, take prior means equal to the true value of parameters with known lower and higher variances.The prior variance indicates the confidence of our prior guess. A large prior variance shows less confidence in prior guess and resulting prior distribution is relatively flat. Likewise, small prior variance indicates greater confidence in prior guess. Here, we took prior variance 0.5 as small and 10 as large, and called these priors as Gamma 2 and Gamma 3, respectively.

For each generated sample, we calculate the proposed estimators discussed in the previous sections. This process is repeated 5000 times, and average estimates (AV) and mean squared errors (MSE) of the estimators are obtained. The simulation results are summarised in Tables \ref{tab:1}-\ref{tab:7}. Studying the results, we reached to the following concussions:
\begin{enumerate}[C1.]
\item As expected, the mean squared error for all estimators, obtained with different parameters values, decreases as sample sizes for stress-strength random variables increase. 
\item Bayes estimator obtained under non-informative prior is equally applicable as maximum likelihood estimator in estimating the parameters since both show approximately same magnitudes for their mean squared errors.
\item In Bayes estimation in case of both symmetric and asymmetric loss functions, the MSE of Bayes estimator obtained under Jeffrey prior is higher than that of obtained under gamma prior, for given values of parameters $(\theta_{1},\theta_{2})$ and sample sizes $(n, m)$.
\item However, Bayes estimators obtained under informative prior show smaller mean squared error than the maximum likelihood estimators.
\item Under gamma prior, it is observed that the MSE of the Bayes estimators increases with increasing prior variance.
\item In general, the MSEs of $(\theta_{1},\theta_{2})$ increase with increasing magnitude of their true values for given (m,n).
\item Bayes estimators obtained under entropy loss function show smaller MSE than Bayes estimators obtained under squared error loss.
\item The MSE of the estimator of stress-strength reliability for small and large values of R, are greater than that obtained for  moderate value (0.5) of R.
\end{enumerate}
\begin{table}[htbp]
  \centering
  \caption{Average estimates (AV) and mean squared errors (MSE) of the estimator R under self and elf with varying n and m, when stress parameter $\theta_{1}$=1 and strength parameter $\theta_{2}$=2 }
    \begin{tabular}{cccccccccc}
    \toprule
    \multirow{2}[4]{*}{(n,m)} & \multirow{2}[4]{*}{MLE} & \multicolumn{2}{c}{\multirow{2}[4]{*}{Jeffrey Prior}} & \multicolumn{6}{c}{Gamma Prior} \\ \cline{5-10}
          &       & \multicolumn{2}{c}{} & \multicolumn{2}{c}{Gamma 1} & \multicolumn{2}{c}{Gamma 2} & \multicolumn{2}{c}{Gamma 3} \\ \cline{2-10}
          & $\hat{R}_{mle}$  & $R^{\star}_{self}$ & $R^{\star}_{elf}$& $R^{\star}_{self}$&$R^{\star}_{elf}$ & $R^{\star}_{self}$&$R^{\star}_{elf}$  & $R^{\star}_{self}$ & $R^{\star}_{elf}$\\ \midrule
    \multirow{2}[4]{*}{(15,20)} & 0.2212 & 0.2284 & 0.2001 & 0.2280 & 0.1997 & 0.2301 & 0.2018 & 0.2282 & 0.1999 \\
          & 0.0080 & 0.0071 & 0.0103 & 0.0072 & 0.0104 & 0.0062 & 0.0095 & 0.0071 & 0.0103 \\
    \multirow{2}[4]{*}{(20,15)} & 0.3630 & 0.3649 & 0.3196 & 0.3651 & 0.3197 & 0.3617 & 0.3173 & 0.3647 & 0.3195 \\
          & 0.0128 & 0.0128 & 0.0063 & 0.0128 & 0.0064 & 0.0112 & 0.0054 & 0.0126 & 0.0063 \\
    \multirow{2}[4]{*}{(30,20)} & 0.1968 & 0.2023 & 0.1838 & 0.2020 & 0.1834 & 0.2035 & 0.1849 & 0.2021 & 0.1836 \\
          & 0.0100 & 0.0090 & 0.0121 & 0.0091 & 0.0121 & 0.0085 & 0.0116 & 0.0090 & 0.0121 \\
    \multirow{2}[4]{*}{(20,30)} & 0.3923 & 0.3927 & 0.3564 & 0.3929 & 0.3565 & 0.3889 & 0.3533 & 0.3925 & 0.3562 \\
          & 0.0165 & 0.0164 & 0.0092 & 0.0165 & 0.0092 & 0.0150 & 0.0083 & 0.0163 & 0.0091 \\
    \multirow{2}[4]{*}{(50,50)} & 0.2857 & 0.2875 & 0.2742 & 0.2875 & 0.2741 & 0.2875 & 0.2741 & 0.2875 & 0.2741 \\
          & 0.0018 & 0.0018 & 0.0018 & 0.0018 & 0.0018 & 0.0017 & 0.0017 & 0.0018 & 0.0018 \\
    \bottomrule
    \end{tabular}%
  \label{tab:1}%
\end{table}%
\begin{table}[htbp]
  \centering
  \caption{Average estimates (AV) and mean squared errors (MSE) of the estimator R under self and elf with varying n and m, when $\theta_{1}$=$\theta_{2}$=1}
    \begin{tabular}{cccccccccc}
        \toprule
    \multirow{2}[4]{*}{(n,m)} & \multirow{2}[4]{*}{MLE} & \multicolumn{2}{c}{\multirow{2}[4]{*}{Jeffrey Prior}} & \multicolumn{6}{c}{Gamma Prior} \\ \cline{5-10}
          &       & \multicolumn{2}{c}{} & \multicolumn{2}{c}{Gamma 1} & \multicolumn{2}{c}{Gamma 2} & \multicolumn{2}{c}{Gamma 3} \\ \cline{2-10}
          & $\hat{R}_{mle}$  & $R^{\star}_{self}$ & $R^{\star}_{elf}$& $R^{\star}_{self}$&$R^{\star}_{elf}$ & $R^{\star}_{self}$&$R^{\star}_{elf}$  & $R^{\star}_{self}$ & $R^{\star}_{elf}$\\ \midrule
    \multirow{2}[4]{*}{(15,20)} & 0.5833 & 0.5792 & 0.5117 & 0.5795 & 0.5119 & 0.5757 & 0.5090 & 0.5792 & 0.5116 \\
          & 0.0139 & 0.0129 & 0.0058 & 0.0129 & 0.0058 & 0.0118 & 0.0054 & 0.0128 & 0.0058 \\
    \multirow{2}[4]{*}{(20,15)} & 0.4166 & 0.4206 & 0.3775 & 0.4203 & 0.3772 & 0.4224 & 0.3790 & 0.4205 & 0.3774 \\
          & 0.0145 & 0.0134 & 0.0211 & 0.0135 & 0.0212 & 0.0127 & 0.0204 & 0.0134 & 0.0211 \\
    \multirow{2}[4]{*}{(30,20)} & 0.3800 & 0.3841 & 0.3561 & 0.3838 & 0.3559 & 0.3856 & 0.3575 & 0.3840 & 0.3560 \\
          & 0.0194 & 0.0183 & 0.0250 & 0.0183 & 0.0251 & 0.0177 & 0.0245 & 0.0183 & 0.0251 \\
    \multirow{2}[4]{*}{(20,30)} & 0.6159 & 0.6119 & 0.5582 & 0.6122 & 0.5585 & 0.6081 & 0.5551 & 0.6118 & 0.5581 \\
          & 0.0182 & 0.0171 & 0.0076 & 0.0172 & 0.0076 & 0.0160 & 0.0070 & 0.0171 & 0.0075 \\
    \multirow{2}[4]{*}{(50,50)} & 0.4996 & 0.4996 & 0.4789 & 0.4996 & 0.4789 & 0.4996 & 0.4789 & 0.4996 & 0.4789 \\
          & 0.0026 & 0.0026 & 0.0029 & 0.0026 & 0.0029 & 0.0025 & 0.0028 & 0.0026 & 0.0029 \\
    \bottomrule
    \end{tabular}%
  \label{tab:2}%
\end{table}%
\begin{table}[htbp]
  \centering
  \caption{Average estimates (AV) and mean squared errors (MSE) of the estimator R under self and elf with varying n and m, when $\theta_{1}$=2, $\theta_{2}$=1}
    \begin{tabular}{cccccccccc}
    \toprule
    \multirow{2}[4]{*}{(n,m)} & \multirow{2}[4]{*}{MLE} & \multicolumn{2}{c}{\multirow{2}[4]{*}{Jeffrey Prior}} & \multicolumn{6}{c}{Gamma Prior} \\ \cline{5-10}
          &       & \multicolumn{2}{c}{} & \multicolumn{2}{c}{Gamma 1} & \multicolumn{2}{c}{Gamma 2} & \multicolumn{2}{c}{Gamma 3} \\ \cline{2-10}
          & $\hat{R}_{mle}$  & $R^{\star}_{self}$ & $R^{\star}_{elf}$& $R^{\star}_{self}$&$R^{\star}_{elf}$ & $R^{\star}_{self}$&$R^{\star}_{elf}$  & $R^{\star}_{self}$ & $R^{\star}_{elf}$\\ \midrule
    \multirow{2}[4]{*}{(15,20)} & 0.7728 & 0.7656 & 0.7027 & 0.7660 & 0.7031 & 0.7541 & 0.6932 & 0.7648 & 0.7021 \\
          & 0.0070 & 0.0062 & 0.0046 & 0.0062 & 0.0045 & 0.0044 & 0.0042 & 0.0060 & 0.0045 \\
    \multirow{2}[4]{*}{(20,15)} & 0.6424 & 0.6403 & 0.5835 & 0.6401 & 0.5835 & 0.6479 & 0.5900 & 0.6409 & 0.5841 \\
          & 0.0120 & 0.0120 & 0.0234 & 0.0120 & 0.0235 & 0.0100 & 0.0211 & 0.0118 & 0.0232 \\
    \multirow{2}[4]{*}{(30,20)} & 0.6111 & 0.6106 & 0.5711 & 0.6104 & 0.5710 & 0.6178 & 0.5775 & 0.6112 & 0.5716 \\
          & 0.0160 & 0.0160 & 0.0256 & 0.0160 & 0.0256 & 0.0141 & 0.0234 & 0.0158 & 0.0254 \\
    \multirow{2}[4]{*}{(20,30)} & 0.7943 & 0.7888 & 0.7429 & 0.7891 & 0.7432 & 0.7770 & 0.7327 & 0.7879 & 0.7422 \\
          & 0.0085 & 0.0077 & 0.0036 & 0.0077 & 0.0036 & 0.0057 & 0.0028 & 0.0075 & 0.0035 \\
    \multirow{2}[4]{*}{(50,50)} & 0.7146 & 0.7127 & 0.6893 & 0.7128 & 0.6893 & 0.7127 & 0.6893 & 0.7128 & 0.6893 \\
          & 0.0018 & 0.0018 & 0.0025 & 0.0018 & 0.0025 & 0.0017 & 0.0024 & 0.0018 & 0.0025 \\
    \bottomrule
    \end{tabular}%
  \label{tab:3}%
\end{table}%
\begin{table}[htbp]
  \centering
  \caption{Average estimates (AV) and mean squared errors (MSE) of the estimator $\theta_{1}$ under self and elf, when $\theta_{1}$=1, with varying n and m}
    \begin{tabular}{crcccccccc}
    \toprule
    \multirow{2}[4]{*}{(n,m)} & \multirow{2}[4]{*}{MLE} & \multicolumn{2}{c}{\multirow{2}[4]{*}{Jeffrey Prior}} & \multicolumn{6}{c}{Gamma Prior} \\ \cline{5-10}
          &       & \multicolumn{2}{c}{} & \multicolumn{2}{c}{Gamma 1} & \multicolumn{2}{c}{Gamma 2} & \multicolumn{2}{c}{Gamma 3} \\ \cline{2-10}
          & \multicolumn{1}{c}{$\hat{\theta}_{1}$} & $\theta^{\star}_{1,self}$   & $\theta^{\star}_{1,elf}$  &$\theta^{\star}_{1,self}$ & $\theta^{\star}_{1,elf}$   & $\theta^{\star}_{1,self}$&$\theta^{\star}_{1,elf}$   & $\theta^{\star}_{1,self}$ & $\theta^{\star}_{1,elf}$ \\ \midrule
    \multirow{2}[4]{*}{(15,20)} & \multicolumn{1}{c}{1.3360} & 1.3360 & 1.2845 & 1.3404 & 1.2887 & 1.3190 & 1.2692 & 1.3383 & 1.2867 \\
          & \multicolumn{1}{c}{0.1867} & 0.1867 & 0.1482 & 0.1905 & 0.1513 & 0.1648 & 0.1306 & 0.1878 & 0.1491 \\
    \multirow{2}[4]{*}{(20,15)} & \multicolumn{1}{c}{0.8064} & 0.8064 & 0.7845 & 0.8078 & 0.7858 & 0.8115 & 0.7894 & 0.8082 & 0.7862 \\
          & \multicolumn{1}{c}{0.0579} & 0.0579 & 0.0656 & 0.0575 & 0.0651 & 0.0552 & 0.0629 & 0.0572 & 0.0649 \\
    \multirow{2}[4]{*}{(30,20)} & \multicolumn{1}{c}{0.7197} & 0.7197 & 0.7068 & 0.7204 & 0.7075 & 0.7239 & 0.7109 & 0.7208 & 0.7078 \\
          & \multicolumn{1}{c}{0.0889} & 0.0889 & 0.0959 & 0.0885 & 0.0955 & 0.0864 & 0.0934 & 0.0883 & 0.0953 \\
    \multirow{2}[4]{*}{(20,30)} & \multicolumn{1}{c}{1.4666} & 1.4666 & 1.4232 & 1.4705 & 1.4269 & 1.4473 & 1.4054 & 1.4682 & 1.4247 \\
          & \multicolumn{1}{c}{0.2833} & 0.2833 & 0.2402 & 0.2875 & 0.2438 & 0.2574 & 0.2181 & 0.2844 & 0.2411 \\
    \multirow{2}[4]{*}{(50,50)} & \multicolumn{1}{c}{1.0108} & 1.0108 & 0.9994 & 1.0117 & 1.0002 & 1.0114 & 0.9999 & 1.0116 & 1.0002 \\
          & \multicolumn{1}{c}{0.0121} & 0.0121 & 0.0117 & 0.0121 & 0.0117 & 0.0118 & 0.0114 & 0.0121 & 0.0117 \\
    \bottomrule
    \end{tabular}%
  \label{tab:4}%
\end{table}%
\begin{table}[htbp]
  \centering
  \caption{Average estimates (AV) and mean squared errors (MSE) of the estimator $\theta_{1}$ under self and elf, when $\theta_{1}$=2, with varying n and m}
    \begin{tabular}{crcccccccc}
    \toprule
    \multirow{2}[4]{*}{(n,m)} & \multirow{2}[4]{*}{MLE} & \multicolumn{2}{c}{\multirow{2}[4]{*}{Jeffrey Prior}} & \multicolumn{6}{c}{Gamma Prior} \\ \cline{5-10}
          &       & \multicolumn{2}{c}{} & \multicolumn{2}{c}{Gamma 1} & \multicolumn{2}{c}{Gamma 2} & \multicolumn{2}{c}{Gamma 3} \\ \cline{2-10}
          & \multicolumn{1}{c}{$\hat{\theta}_{1}$} & $\theta^{\star}_{1,self}$   & $\theta^{\star}_{1,elf}$  &$\theta^{\star}_{1,self}$ & $\theta^{\star}_{1,elf}$   & $\theta^{\star}_{1,self}$&$\theta^{\star}_{1,elf}$   & $\theta^{\star}_{1,self}$ & $\theta^{\star}_{1,elf}$ \\ \midrule
    \multirow{2}[4]{*}{(15,20)} & \multicolumn{1}{c}{2.6947} & 2.6947 & 2.5770 & 2.7066 & 2.5879 & 2.4967 & 2.4157 & 2.6856 & 2.5689 \\
          & \multicolumn{1}{c}{0.8359} & 0.8359 & 0.6503 & 0.8562 & 0.6662 & 0.3839 & 0.3245 & 0.7997 & 0.6208 \\
    \multirow{2}[4]{*}{(20,15)} & \multicolumn{1}{c}{1.6242} & 1.6242 & 1.5753 & 1.6287 & 1.5795 & 1.6598 & 1.6102 & 1.6318 & 1.5825 \\
          & \multicolumn{1}{c}{0.2334} & 0.2334 & 0.2661 & 0.2307 & 0.2632 & 0.1924 & 0.2233 & 0.2267 & 0.2591 \\
    \multirow{2}[4]{*}{(30,20)} & \multicolumn{1}{c}{1.4482} & 1.4482 & 1.4195 & 1.4508 & 1.4219 & 1.4810 & 1.4517 & 1.4538 & 1.4248 \\
          & \multicolumn{1}{c}{0.3506} & 0.3506 & 0.3810 & 0.3480 & 0.3784 & 0.3122 & 0.3413 & 0.3443 & 0.3747 \\
    \multirow{2}[4]{*}{(20,30)} & \multicolumn{1}{c}{2.9495} & 2.9495 & 2.8503 & 2.9594 & 2.8596 & 2.7372 & 2.6696 & 2.9372 & 2.8389 \\
          & \multicolumn{1}{c}{1.2273} & 1.2273 & 1.0231 & 1.2488 & 1.0412 & 0.6943 & 0.6091 & 1.1855 & 0.9875 \\
    \multirow{2}[4]{*}{(50,50)} & \multicolumn{1}{c}{2.0296} & 2.0296 & 2.0038 & 2.0321 & 2.0062 & 2.0289 & 2.0031 & 2.0318 & 2.0059 \\
          & \multicolumn{1}{c}{0.0546} & 0.0546 & 0.0521 & 0.0550 & 0.0523 & 0.0488 & 0.0466 & 0.0543 & 0.0518 \\
    \bottomrule
    \end{tabular}%
  \label{tab:5}%
\end{table}%
\begin{table}[htbp]
  \centering
  \caption{Average estimates (AV) and mean squared errors (MSE) of the estimator $\theta_{2}$ under self and elf, when $\theta_{2}$=1 , with varying n and m}
    \begin{tabular}{crcccccccc}
    \toprule
   \multirow{2}[4]{*}{(n,m)} & \multirow{2}[4]{*}{MLE} & \multicolumn{2}{c}{\multirow{2}[4]{*}{Jeffrey Prior}} & \multicolumn{6}{c}{Gamma Prior} \\ \cline{5-10}
          &       & \multicolumn{2}{c}{} & \multicolumn{2}{c}{Gamma 1} & \multicolumn{2}{c}{Gamma 2} & \multicolumn{2}{c}{Gamma 3} \\ \cline{2-10}
          & \multicolumn{1}{c}{$\hat{\theta}_{2}$} & $\theta^{\star}_{2,self}$   & $\theta^{\star}_{2,elf}$  &$\theta^{\star}_{2,self}$ & $\theta^{\star}_{2,elf}$   & $\theta^{\star}_{2,self}$&$\theta^{\star}_{2,elf}$   & $\theta^{\star}_{2,self}$ & $\theta^{\star}_{2,elf}$ \\ \midrule
    \multirow{2}[4]{*}{(15,20)} & \multicolumn{1}{c}{1.0291} & 1.0291 & 1.0003 & 1.0312 & 1.0023 & 1.0293 & 1.0005 & 1.0310 & 1.0021 \\
          & \multicolumn{1}{c}{0.0330} & 0.0330 & 0.0301 & 0.0333 & 0.0303 & 0.0310 & 0.0283 & 0.0331 & 0.0301 \\
    \multirow{2}[4]{*}{(20,15)} & \multicolumn{1}{c}{1.0409} & 1.0409 & 1.0023 & 1.0438 & 1.0050 & 1.0402 & 1.0017 & 1.0435 & 1.0047 \\
          & \multicolumn{1}{c}{0.0465} & 0.0465 & 0.0410 & 0.0472 & 0.0415 & 0.0424 & 0.0375 & 0.0467 & 0.0410 \\
    \multirow{2}[4]{*}{(30,20)} & \multicolumn{1}{c}{1.0343} & 1.0343 & 1.0053 & 1.0365 & 1.0074 & 1.0344 & 1.0054 & 1.0363 & 1.0072 \\
          & \multicolumn{1}{c}{0.0344} & 0.0344 & 0.0312 & 0.0348 & 0.0314 & 0.0323 & 0.0293 & 0.0346 & 0.0312 \\
    \multirow{2}[4]{*}{(20,30)} & \multicolumn{1}{c}{1.0236} & 1.0236 & 1.0043 & 1.0250 & 1.0057 & 1.0241 & 1.0048 & 1.0249 & 1.0056 \\
          & \multicolumn{1}{c}{0.0211} & 0.0211 & 0.0197 & 0.0213 & 0.0198 & 0.0203 & 0.0189 & 0.0212 & 0.0197 \\
    \multirow{2}[4]{*}{(50,50)} & \multicolumn{1}{c}{1.0123} & 1.0123 & 1.0008 & 1.0132 & 1.0017 & 1.0128 & 1.0014 & 1.0131 & 1.0016 \\
          & \multicolumn{1}{c}{0.0125} & 0.0125 & 0.0121 & 0.0126 & 0.0121 & 0.0123 & 0.0118 & 0.0125 & 0.0121 \\
    \bottomrule
    \end{tabular}%
  \label{tab:6}%
\end{table}%
\begin{table}[htbp]
  \centering
  \caption{Average estimates (AV) and mean squared errors (MSE) of the estimator $\theta_{2}$ under self and elf, when $\theta_{2}$=2, with varying n and m}
    \begin{tabular}{crcccccccc}
    \toprule
   \multirow{2}[4]{*}{(n,m)} & \multirow{2}[4]{*}{MLE} & \multicolumn{2}{c}{\multirow{2}[4]{*}{Jeffrey Prior}} & \multicolumn{6}{c}{Gamma Prior} \\ \cline{5-10}
          &       & \multicolumn{2}{c}{} & \multicolumn{2}{c}{Gamma 1} & \multicolumn{2}{c}{Gamma 2} & \multicolumn{2}{c}{Gamma 3} \\ \cline{2-10}
          & \multicolumn{1}{c}{$\hat{\theta}_{2}$} & $\theta^{\star}_{2,self}$   & $\theta^{\star}_{2,elf}$  &$\theta^{\star}_{2,self}$ & $\theta^{\star}_{2,elf}$   & $\theta^{\star}_{2,self}$&$\theta^{\star}_{2,elf}$   & $\theta^{\star}_{2,self}$ & $\theta^{\star}_{2,elf}$ \\ \midrule
    \multirow{2}[4]{*}{(15,20)} & \multicolumn{1}{c}{2.1147} & 2.1147 & 2.0264 & 2.1235 & 2.0345 & 2.0787 & 1.9979 & 2.1190 & 2.0304 \\
          & \multicolumn{1}{c}{0.2382} & 0.2382 & 0.2038 & 0.2427 & 0.2064 & 0.1366 & 0.1255 & 0.2305 & 0.1965 \\
    \multirow{2}[4]{*}{(20,15)} & \multicolumn{1}{c}{2.0723} & 2.0723 & 2.0071 & 2.0787 & 2.0131 & 2.0571 & 1.9944 & 2.0766 & 2.0111 \\
          & \multicolumn{1}{c}{0.1571} & 0.1571 & 0.1407 & 0.1593 & 0.1419 & 0.1095 & 0.1010 & 0.1538 & 0.1373 \\
    \multirow{2}[4]{*}{(30,20)} & \multicolumn{1}{c}{2.0787} & 2.0787 & 2.0132 & 2.0852 & 2.0193 & 2.0622 & 1.9994 & 2.0829 & 2.0172 \\
          & \multicolumn{1}{c}{0.1636} & 0.1636 & 0.1460 & 0.1660 & 0.1473 & 0.1132 & 0.1041 & 0.1602 & 0.1424 \\
    \multirow{2}[4]{*}{(20,30)} & \multicolumn{1}{c}{2.0465} & 2.0465 & 2.0033 & 2.0508 & 2.0073 & 2.0417 & 1.9991 & 2.0498 & 2.0065 \\
          & \multicolumn{1}{c}{0.0984} & 0.0984 & 0.0914 & 0.0993 & 0.0920 & 0.0800 & 0.0749 & 0.0973 & 0.0902 \\
    \multirow{2}[4]{*}{(50,50)} & \multicolumn{1}{c}{2.0262} & 2.0262 & 2.0004 & 2.0287 & 2.0028 & 2.0256 & 1.9999 & 2.0284 & 2.0025 \\
          & \multicolumn{1}{c}{0.0555} & 0.0555 & 0.0531 & 0.0558 & 0.0533 & 0.0496 & 0.0475 & 0.0551 & 0.0527 \\
    \bottomrule
    \end{tabular}%
  \label{tab:7}%
\end{table}%
\section{Real data application}
In this section, we considered two real data sets initially reported by \cite{Efron88}, to demonstrate that the methodologies proposed in the previous sections, can be used in practice. The data sets represent the survival times of two groups of patients suffering from Head and Neck cancer disease. The patients in one group were treated using radiotherapy (RT) whereas the patients belonging to other group were treated using a combined radiotherapy and chemotherapy (CT+RT). \cite{Makkar14} showed how a lifetime model with upside-down bathtub hazard rate is useful to modelling such problem and analysed these data sets under Bayesian paradigm using log-normal lifetime model. The data sets are as follows:
\begin{align*}
\textbf{Data (X):~} &6.53,7,10.42,14.48,16.10,22.70,34,41.55,42,45.28,49.40,53.62,63,64,83,84,91,108,112,\\
&129,133,133,139,140,140,146,149,154,157,160,160,165,146,149,154,157,160,160,165,\\
&173,176,218,225,241,248,273,277,297,405,417,420,440,523,583,594,1101,1146,1417	
\end{align*}
\begin{align*}
\textbf{Data (Y):~}&12.20,23.56,23.74,25.87,31.98,37,41.35,47.38,55.46,58.36,63.47,68.46,78.26,74.47,\\
&81.43,84,92,94,110,112,119,127,130,133,140,146,155,159,173,179,194,195,209, \\
&249,281,319,339,432,469,519,633,725,817,1776
\end{align*}
\begin{figure}
\centering
\includegraphics[height=10cm,width=10cm]{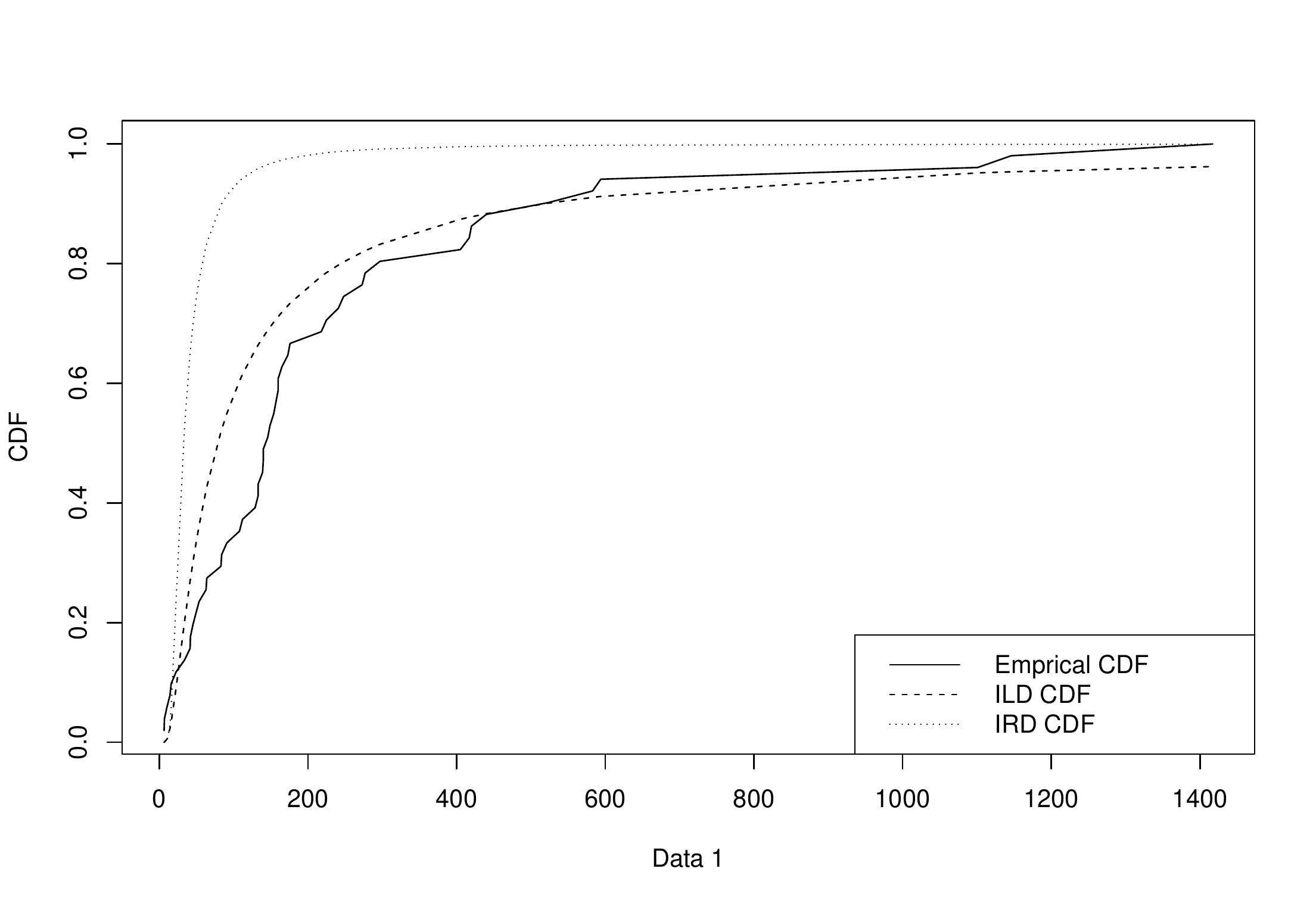}\\
\caption{Empirical and fitted cumulative density function for Data 1}
\label{fig3}
\end{figure}
\begin{figure}
\centering
\includegraphics[height=10cm,width=10cm]{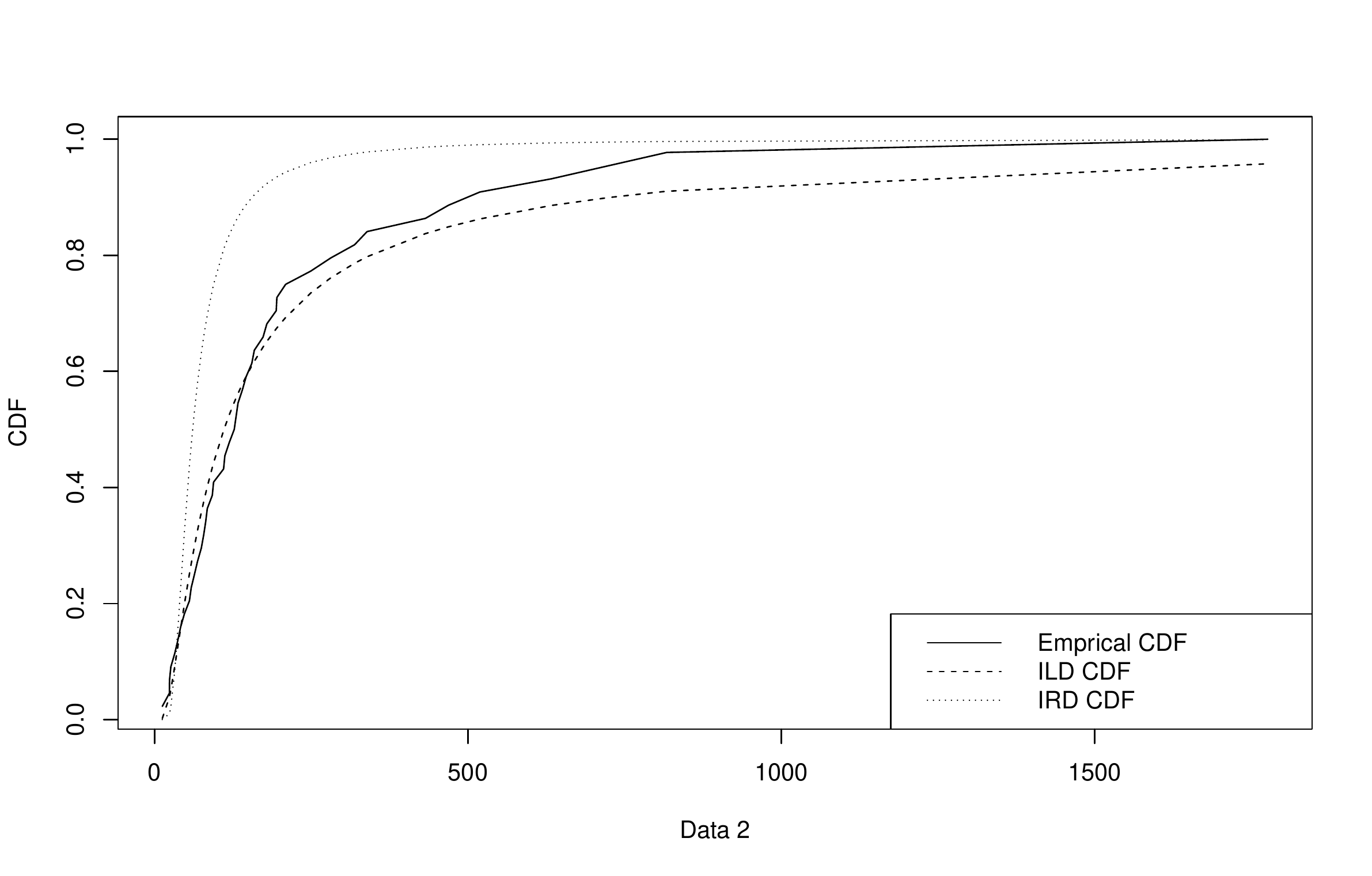}\\
\caption{Empirical and fitted cumulative density function for Data  2}
\label{fig4}
\end{figure}
First we checked the validity of the inverse Lindley distribution for given data sets by using Kologorov-Smirnov goodness-of-fit test(K-S), Akaike information criterion (AIC) and Bayes information criterion (BIC).

We compared the applicability of inverse Lindely distribution with a competing one parameter distribution, inverse Rayleigh distribution(IRD). Based on real data sets, all statistics mentioned above are computed and shown for both distributions in Table \ref{tab:8}. Figures (\ref{fig3}) and (\ref{fig4}) show the fitted and empirical Cdfs for data set 1 and 2 respectively. Clearly, the inverse Lindley distribution fits well to the data sets in comparison to inverse Rayleigh distribution.

The maximum likelihood and various Bayes estimates of $\theta_{1},~ \theta_{2}$ and $R$ are summarised in Table \ref{tab:9}. Since, we dint have any prior information besides few observations, we used only non-informative priors, Jeffrey and Gamma 1 priors in real data study.
\begin{table}[htbp]
  \centering
  \caption{The model fitting summary for both the data sets.}
    \begin{tabular}{cclcccc}
    \toprule
    \multicolumn{1}{c}{} & \multicolumn{1}{c}{Model} & \multicolumn{1}{c}{MLE} & \multicolumn{1}{c}{Log-Likelihood} & \multicolumn{1}{c}{AIC} & \multicolumn{1}{c}{BIC} & \multicolumn{1}{c}{K-S} \\
        \midrule
    \multicolumn{1}{c}{\multirow{2}[4]{*}{Data 1}} & \multicolumn{1}{c}{ILD} & \multicolumn{1}{c}{55.4540} & \multicolumn{1}{c}{-340.7515} & \multicolumn{1}{c}{683.5029} & \multicolumn{1}{c}{685.4347} & \multicolumn{1}{c}{0.2634} \\
    \multicolumn{1}{c}{} & \multicolumn{1}{c}{IRD} & \multicolumn{1}{c}{741.3652} & \multicolumn{1}{c}{-419.0671} & \multicolumn{1}{c}{840.1341} & \multicolumn{1}{c}{842.0660} & \multicolumn{1}{c}{0.6039} \\
            \midrule
    \multicolumn{1}{c}{\multirow{2}[4]{*}{Data 2}} & \multicolumn{1}{c}{ILD} & \multicolumn{1}{c}{77.6755} & \multicolumn{1}{c}{-344.1048} & \multicolumn{1}{c}{690.2096} & \multicolumn{1}{c}{691.9938} & \multicolumn{1}{c}{0.0799} \\
    \multicolumn{1}{c}{} & \multicolumn{1}{c}{IRD} & \multicolumn{1}{c}{2547.4170} & \multicolumn{1}{c}{-480.3576} & \multicolumn{1}{c}{962.7151} & \multicolumn{1}{c}{964.4993} & \multicolumn{1}{c}{0.3783} \\
    \bottomrule
    \end{tabular}%
  \label{tab:8}%
\end{table}%
\begin{table}[htbp]
  \centering
  \caption{The maximum likelihood and Bayes estimates of $\theta_{1}, \theta_{2}$ and $R$ based on real data sets.}
    \begin{tabular}{cccccc} 
    \toprule
    \multirow{2}[4]{*}{Parameter} & \multirow{2}[4]{*}{MLE} & \multicolumn{2}{c}{Jeffery} & \multicolumn{2}{c}{Gamma 1} \\ \cline{3-6}
          &       & SELF  & ELF   & SELF  & ELF \\     \midrule
    R     & 0.5847 & 0.5834 & 0.5737 & 0.5834 & 0.5736 \\
    $\theta_{1}$    & 77.6755 & 77.6755 & 75.9910 & 77.6963 & 76.0109 \\
    $\theta_{2}$     & 55.4540 & 55.4540 & 54.4230 & 55.4713 & 54.4398 \\
    \bottomrule
    \end{tabular}%
  \label{tab:9}%
\end{table}%
\section{Conclusion}
In this paper, we purposed an inverse Lindley distribution and studied its mathematical and statistical properties such as quantile, mode, stochastic ordering. The flexibility of the propose model has been demonstrated by showing the behaviours of the Pdf and hazard functions mathematically as well as graphically. It has been found that the inverse Lindley distribution fits quit well the data of survival of Head and Neck cancer patients. Further, we addressed the problem of estimating the stress-strength reliability $R$=$P[Y<X]$ where $X$ and $Y$ follow the inverse Lindley distributions. The estimation of the parameters and stress-strength reliability $R$ have been approached by both classical and Bayesian methods. We could use the asymptotic distribution of $R$ to construct confidence intervals. We computed the Bayes estimators of $R$ based on the Jeffrey and gamma priors using a squared error loss and an entropy loss functions. Since the Bayes estimators cannot be obtained in explicit form, we proposed the use of the Lindley's approximation to compute the Bayes estimates. By simulation study, we concluded that Bayes estimators obtained under informative priors perform well and have smaller mean squared error as compared to that of obtained under the non-informative priors, and maximum likelihood estimators. We could also observed that entropy loss function is a suitable loss function than squared error loss. Further extensions of this distribution are under progress and evolved versions having effective shape parameter(s) will be presented in the literature very soon. Finally, we would recommend the use of inverse Lindely distribution as a reliability and lifetime model to modelling the real problems encountered in engineering, medical science, biological science and other applied sciences.
\section*{Appendix}
\begin{equation*}
\begin{split}
&L_{111}=\frac{4n}{\theta_{1}^{3}}-\frac{2n}{\left(1+\theta_{1}\right)^{3}};L_{222}=\frac{4m}{\theta_{2}^{3}}-\frac{2m}{\left(1+\theta_{2}\right)^{3}};L_{112}=L_{122}=0\\
& \begin{pmatrix}
\sigma_{11}& \sigma_{12}\\
\sigma_{21}& \sigma_{22}
\end{pmatrix}
=-\begin{pmatrix}
-\frac{2n}{\theta_{2}^{2}}+\frac{n}{\left(1+\theta_{2}\right)^{2}}& 0 \\
0 & -\frac{2m}{\theta_{2}^{2}}+\frac{m}{\left(1+\theta_{2}\right)^{2}}
\end{pmatrix}^{-1}\\
&\rho_{1}=\frac{\theta_{1}+2}{\theta_{1}^{2}+4\theta_{1}+2}-\frac{1}{\theta_{1}}-\frac{1}{\left(1+\theta_{1}\right)}, \rho_{2}=\frac{\theta_{2}+2}{\theta_{2}^{2}+4\theta_{2}+2}-\frac{1}{\theta_{2}}-\frac{1}{\left(1+\theta_{2}\right)}\\
&R=\frac{\theta_{1}^{2}\left[ \left(\theta_{1}+\theta_{2} \right)^{2}\left(1+\theta_{2} \right)+\left( \theta_{1}+\theta_{2}\right)\left( 1+2\theta_{2}\right)+2\theta_{2} \right] }{\left(1+\theta_{1} \right)\left( 1+\theta_{2}\right)\left( \theta_{1}+\theta_{2}\right)^{3}}=\frac{\delta}{\lambda}\\
&R_{1}=\frac{\lambda\delta_{1}-\delta\lambda_{1}}{\lambda^{2}}, R_{2}=\frac{\lambda\delta_{2}-\delta\lambda_{2}}{\lambda^{2}}, R_{11}=\frac{\lambda^{2}\left(\lambda\delta_{11}-\delta\lambda_{11}\right)-2\lambda\lambda_{1}\left(\lambda\delta_{1}-\delta\lambda_{1} \right)}{\lambda^{4}}\\
&R_{22}=\frac{\lambda^{2}\left(\lambda\delta_{22}-\delta\lambda_{22}\right)-2\lambda\lambda_{2}\left(\lambda\delta_{2}-\delta\lambda_{2} \right)}{\lambda^{4}}, R_{1}^{-1}=\frac{\delta\lambda_{1}-\lambda\delta_{1}}{\delta^{2}}, R_{2}^{-1}=\frac{\delta\lambda_{2}-\lambda\delta_{2}}{\delta^{2}}\\
&R_{11}^{-1}=\frac{\delta^{2}\left(\delta\lambda_{11}-\lambda\delta_{11}\right)-2\delta\delta_{1}\left(\delta\lambda_{1}-\lambda\delta_{1} \right)}{\delta^{4}}, R_{22}^{-1}=\frac{\delta^{2}\left(\delta\lambda_{22}-\lambda\delta_{22}\right)-2\delta\delta_{2}\left(\delta\lambda_{2}-\lambda\delta_{2} \right)}{\delta^{4}}\\
\end{split}
\end{equation*}
\begin{equation*}
\begin{split}
&\lambda= \left(1+\theta_{1} \right)\left( 1+\theta_{2}\right)\left( \theta_{1}+\theta_{2}\right)^{3}\\
&\delta=\theta_{1}^{2}\left[ \left(\theta_{1}+\theta_{2} \right)^{2}\left(1+\theta_{2} \right)+\left( \theta_{1}+\theta_{2}\right)\left( 1+2\theta_{2}\right)+2\theta_{2} \right]\\
&\lambda_{1}= \left( 1+\theta_{2}\right)\left( \theta_{1}+\theta_{2}\right)^{3}+3\left(1+\theta_{1} \right)\left( 1+\theta_{2}\right)\left( \theta_{1}+\theta_{2}\right)^{2}\\
&\lambda_{2}= \left( 1+\theta_{1}\right)\left( \theta_{1}+\theta_{2}\right)^{3}+3\left(1+\theta_{1} \right)\left( 1+\theta_{2}\right)\left( \theta_{1}+\theta_{2}\right)^{2}\\
&\delta_{1}=2\theta_{1}\left[ \left(\theta_{1}+\theta_{2} \right)^{2}\left(1+\theta_{2} \right)+\left( \theta_{1}+\theta_{2}\right)\left( 1+2\theta_{2}\right)+2\theta_{2} \right]+\theta_{1}^{2}\left[ 2\left(\theta_{1}+\theta_{2} \right)\left(1+\theta_{2}\right)+\left( 1+2\theta_{2}\right) \right]\\
&\delta_{2}=\theta_{1}^{2}\left[\left(\theta_{1}+\theta_{2} \right)^{2}+2\left( \theta_{1}+\theta_{2}\right)\left( 1+\theta_{2}\right)+2\left(\theta_{1}+\theta_{2} \right)+\left( 1+2\theta_{2}\right)+2\right]\\
&\lambda_{11}=6\left( 1+\theta_{2}\right)\left( \theta_{1}+\theta_{2}\right)\left(2\theta_{1}+\theta_{2}+1 \right) \\
&\lambda_{22}=6\left( 1+\theta_{1}\right)\left( \theta_{1}+\theta_{2}\right)\left(2\theta_{2}+\theta_{1}+1 \right) \\
&\delta_{11}=2\left(\theta_{1}+\theta_{2}\right)\left(1+\theta_{2}\right)\left(5\theta_{1}+\theta_{2}\right)+2\left( 3\theta_{1}+\theta_{2}\right)\left( 1+2\theta_{2}\right)+2\theta_{2}(\theta_{1}^{2}+2)+2\theta_{1}^{2}\\
&\delta_{22}=\theta_{1}^{2}\left[ 2\left(1+\theta_{2}\right)+ 4\left( \theta_{1}+\theta_{2}\right)+ 4\right]
\end{split}
\end{equation*}
\bibliography{REFERENCES}
\end{document}